\documentclass{fundam}

\usepackage{amssymb}
\usepackage{amsmath}
\usepackage[english]{babel}
\usepackage{hyperref}

\begin{document}


\setcounter{page}{1}
\publyear{24}
\papernumber{2170}
\volume{191}
\issue{1}

\finalVersionForARXIV


\title{Link Residual Closeness of Harary Graphs}

\author{Chavdar Dangalchev\thanks{Address for correspondence:  Bulgarian Academy of Sciences,
                           Institute of Mathematics and Informatics, Bulgaria.  \newline \newline
                    \vspace*{-6mm}{\scriptsize{Received May 2023; \ accepted December 2023.}}}
\\
 Bulgarian Academy of Sciences\\
 Institute of Mathematics and Informatics, Bulgaria\\
  dangalchev@hotmail.com}

\maketitle

\runninghead{Ch. Dangalchev}{Link Residual Closeness of Harary Graphs}

\begin{abstract}
The study of networks characteristics is an important subject in different fields, like
math, chemistry, transportation, social network analysis etc.
The residual closeness is one of the most sensitive measure
of graphs' vulnerability. In this article we calculate
the link residual closeness of Harary graphs.
\end{abstract}

\section{Introduction}

One important characteristic of networks is their robustness, studied in many different
fields of the science.
One of the most sensitive measures of network's vulnerability is residual closeness,
introduced in  [1] - Dangalchev
proposed to measure the closeness of a graph after removing a vertex or a link (edge).
The  definition for the closeness of a simple undirected graph, introduced in [1], is:
\[
C(G)=\sum\limits_i  \sum\limits_{j\ne i} 2^{-d(i,j)}.
\]
In the above formula,  $d(i,j)$  is the standard distance between vertices $i$  and $j$.
The advantages of the above definition are that it can be used
 for not connected graphs and it is convenient for creating formulae for graph operations.

\medskip
Let $r$   and $s$  be a pair of connected vertices of graph  $G$ and graph  $G_{r,s}$  be the graph, constructed by removing link $(r,s)$ from graph  $G$.
Let  $d' (i,j)$  be the distance between vertices
$i$   and $j$  in  graph  $G_{r,s}$.
Using the above formula, with distances $d'(i,j)$ instead of $d(i,j)$, we can calculate the closeness of graph $G_{r,s}$.
The link residual closeness R is defined in [1] as:
\[
R(G)={\min_{r,s} \{ C(G_{r,s}) \}  }.
\]
If we remove a vertex, instead of a link, we can define vertex residual closeness.
The vertex residual closeness is more important for the social network analysis, while the link residual closeness is studied in transportation, utility networks, etc.
In this article we will
consider only the link residual closeness.
To find the difference between the closeness and the residual closeness we have to compare  distances $d(i,j)$ and $d'(i,j)$.

\medskip
Harary graphs are introduced in [2] by F. Harary as
 graphs that are $k$-connected, having $n$ vertices with the least number of edges.
The notation $H_{k,n}$ for  Harary graphs,
where $2 \le k < n$ is used in West [3].
A simple construction of  Harary graphs is:
Let us place $n$ vertices
in a circle and name them $1,2,3,...,n$.
In case of $k=2p$ even, every vertex is connected to
nearest $p$ vertices in each direction.
In case of $k=2p+1$ odd and $n=2q$ even, $H_{k,n}$
is created by connecting every vertex to the nearest $p$ vertices in each direction
and to the diametrically opposite vertex (adding links $(i,i+q)$).
In these two cases there is an automorphism
between any two vertices.
In case of $k=2p+1$ odd and $n=2q+1$ odd, the Harary graph is created by
connecting every vertex to the nearest $p$ vertices in each direction
and for vertices $i \in [1,q+1]$ are added links $(i,i+q)$.
This way every vertex is connected to $k=2p+1$ other vertices, except for vertex $q+1$, which is connected to $2p+2$ vertices: in addition to the
$2p$ links to the neighbors, there are 2 more links -  $(1,q+1)$ and  $(q+1,2q+1)$.

The relative impact of a failure of a link can be seen in normalized residual closeness
NR ([1]) of graph $G$: $NR(G) = (C(G)-R(G)) \symbol{92} C(G)$.
In this article we will calculate the difference between the closeness and
the link residual closeness of Harary graphs.
The closeness and the vertex residual closeness of some Harary graphs are given in [4].
We can determine the link residual closeness using the results of this article and the closeness from [4]. Throughout this article we will use the term “residual closeness” instead of “link residual closeness”.
More information on closeness, residual closeness, and additional closeness can be found in [5-25].

\section{Residual closeness of $H_{2,n}$}

Graph $H_{2,n}$ is cycle graph $C_n$.
After deleting any link of $H_{2,n}$ we receive path graph $P_n$.
Using formulae for closenesses of cycle graphs (given in [4]) and path graphs (in [1]) we can prove:
\begin{theorem}
The residual closeness of Harary graph $H_{2,n}$ is:
\[
R(H_{2,2k}) = C(H_{2,2k}) - 4 + 2^ {2-2k} + 3k2^ {1-k},
\]
\[
R(H_{2,2k+1}) = C(H_{2,2k+1}) - 4 + 2^ {1-2k} + (2k+1)2^ {1-k}.
\]
\end{theorem}
\begin{proof}
The formulae  for closeness of cycle graphs, given in [4] are:
\[
C(C_{2k}) = 4k - 6k2^ {-k},
\]
\[
C(C_{2k+1}) = 2(2k+1) - 2(2k+1)2^ {-k}.
\]
The formula for closeness of path graphs, given in [1] is:
\[
C(P_n)  = 2n - 4 + 2^ {2-n}.
\]
Replacing in the last formula $n$ with $2k$ and $2k+1$,
and subtracting it from the upper two formulae, we prove the theorem.
\end{proof}

\section{Residual closeness of $H_{2p,n}$}

We will consider all cases where $p>1$.  In graph $H_{2p,n}$ vertex $1$ is connected to vertices
$2$,...,$p+1$ as well as to $n$,...,$n-p+1$. Because of the automorphism between any two vertices of the graph we will consider only deleting links starting from vertex $1$.

\medskip
By deleting link $(1,2)$, distance $d(1,2)$ is changed from $1$ to $2$. The new distance is $d'(1,2)=d(1,3)+d(3,2)=2$. The same is the change of the distances  (from $1$ to $2$) when deleting links $(1,3)$,...,$(1,p+1)$, because
$d'(1,j)=d(1,2)+d(2,j)=2$. No other distances are changed when
$n \le 4p$. Every change of a distance should be counted twice,
e.g. for distance $d(1,2)$ and for distance $d(2,1)$.
In this case the difference $\Delta_1$ between the
 closeness and the residual closeness is
 $\Delta_1 = 2 \cdot 2^{-1} - 2 \cdot 2^{-2}=0.5$ and:
\[
R(H_{2p,n}) =  C(H_{2p,n}) - 0.5,  \quad   n \le 4p.
\]

Deleting links $(1,2)$,...,$(1,p)$ cannot result in any changes between different vertices.
For example, if $i,...1,s,t,...j$ is a path with the shortest distance between vertices $i$ and $j$, where $s \in [2,p]$, then the same distance is given by path $i,...1,s+1,t,...j$.

\medskip
When $n = 4p +1$ deleting link $(1,p+1)$ will change,
in addition to distance $d(1,p+1)$,  also distance $d(1,2p+1)$ from 2 to 3.  The same will be the change for distance $d(n-p+1,p+1)$.
Deleting any other link will not have bigger change in closenesses.
The new difference is $\Delta_2 = \Delta_1 + 2(2 \cdot 2^{-2} - 2 \cdot 2^{-3})= 1$.
The same ($\Delta_2$) is the difference  when $n = 4p +2,...,6p$.

When $n = 6p +1$ deleting link $(1,p+1)$ will change additionally distance
$d(1,3p+1)$ from 3 to 4.  The same will be the change for 2 other  distances:
$d(n-p+1,2p+1)$ and $d(n-2p+1,p+1)$.
The new difference $\Delta_3 = \Delta_2 + 3(2 \cdot 2^{-3} - 2 \cdot 2^{-4})= \Delta_2 + 3 \cdot 2^{-3} = 1.375$.
The same ($\Delta_3$) is the difference when $n = 6p +2,...,8p$.
Using the floor function  $c=\lfloor \frac{a}{b} \rfloor$, where $c$ is
the integer part of the division of $a$ by $b$, we can prove:

\begin{theorem}
The residual closeness of Harary graph $H_{2p,n}$ is:
\[
R(H_{2p,n}) =  C(H_{2p,n}) - 2 + (k+2)2^{-k},
\]
where $k = \lfloor \frac{n-1}{2p} \rfloor$ and $p>1$.
\end{theorem}
\begin{proof}
In general:
\[
\Delta_k = \Delta_{k-1} + k2^{-k}=2^{-1}+2 \cdot 2^{-2}+3 \cdot 2^{-3}+...+ k2^{-k}.
\]
in Appendix A is proven formula(1):
\begin{equation}
\label{eq1}
3 \cdot 2^{-2}+...+k2^{1-k} = 2 -(k+2)2^{1-k}.
\end{equation}
Dividing formula (1) by 2 and adding 1 we receive: $\Delta_k = 2 - (k+2)2^{-k}$,
which proves the theorem.
\end{proof}

\section{Residual closeness of $H_{3,2n}$}

There is automorphism
between any two vertices of graph $H_{3,2n}$ - instead of deleting link $(i,i+1)$
we will delete link $(1,2)$;  instead of deleting link $(i,i+n)$
we will delete link $(1,n+1)$;

\medskip
$H_{3,4}$ is a complete graph and
$d'(1,2)= d(1,3) + d(3,2) = 2$. No other distances are changed.
We have $\Delta_2 = 2 \cdot 2^{-1} - 2 \cdot 2^{-2}= 0.5$ and:
\[
R(H_{3,4})=C(H_{3,4}) - 0.5.
\]
For $n \ge 3$ we have to consider 2 cases.

\medskip
\textbf {Case 1 - Deleting link $(1,n+1)$:}

 \noindent Distance $d(1,n+1)$ is changed from 1 to 3:
\[
d'(1,n+1)=d(1,2)+d(2,n+2)+d(n+2,n+1) = 3.
\]
This is the only changed distance. For example: $d(1,n+2)= d(1,n+1)+ d(n+1,n+2)$ and $d’(1,n+2)= d(1,2)+ d(2,n+2)$.
The difference in closenesses is: $\Delta = 2(2^{-1} -  2^{-3})= 0.75$.

\medskip
\textbf {Case 2 -  Deleting link $(1,2)$:}\smallskip

 \noindent \quad \quad \quad \textbf {A)} Distance $d(1,2)$ is changed from 1 to 3:
\[
d'(1,2)=d(1,n+1)+d(n+1,n+2)+d(n+2,2) = 3.
\]
This is true when $n \ge 3$.
When $n =3$ this is the only changed distance, hence:
\[
\Delta_3 = 2 \cdot 2^{-1} -  2 \cdot 2^{-3}= 0.75,
\]
\[
R(H_{3,6})=C(H_{3,6}) - \Delta_3  = C(H_{3,6}) - 0.75.
\]

 \noindent \quad \quad \quad \textbf {B)} When $n =4$ two more distances are changed from 2 to 3. Distance: $d'(1,3)=d(1,5)+d(5,4)+d(4,3)=3$. The same is the situation with distance $d'(2,8)$, hence:
\[
\Delta_4=\Delta_3 + 2(2 \cdot 2^{-2} -  2 \cdot 2^{-3})= 1.25.
\]
\[
R(H_{3,8})=C(H_{3,8}) - \Delta_4 = C(H_{3,8}) - 1.25.
\]

 \noindent \quad \quad \quad \textbf {C)} When $n \ge 5$,
 distance $d(1,3)$ is changed from 2 to 4:
\[
d'(1,3)=d(1,n+1)+d(n+1,n+2)+d(n+2,n+3)+d(n+3,3) = 4.
\]
The same is situation with distance $d(2,2n)$.
When $n=5$ these are the only changed distances and:
\[
\Delta_5=\Delta_3 + 2(2 \cdot 2^{-2} -  2 \cdot 2^{-4})= 1.5,
\]
\[
R(H_{3,10})=C(H_{3,10}) - \Delta_5= C(H_{3,10}) - 1.5.
\]

 \noindent \quad \quad \quad \textbf {D)} In general, when $n =2k$ distance $d(1,k+1)$ is changed from $k$ to $k+1$:
\[
d'(1,k+1)=d(1,n+1)+d(n+1,2k)+...+d(k+2,k+1)=k+1,
\]
or the closeness is changed with $\Delta = 2 \cdot 2^{-k} -  2 \cdot 2^{-k-1}= 2^{-k}$.
The same is true for other $k-1$ distances:
$d(2n,k)$,$d(2n-1,k-1)$,...,$d(2n-k+2,2)$.
The difference in closenesses is:
\[
 \Delta_{2k} =\Delta_{2k-1} + k2^{-k}.
\]
The residual closeness is:
\[
R(H_{3,4k})=C(H_{3,4k}) - \Delta_{2k}= C(H_{3,4k}) - \Delta_{2k-1} - k2^{-k}.
\]

 \noindent \quad \quad \quad \textbf {E)} Distance $d(1,k+1)$, when $n \ge 2k+1$, is changed from $k$ to $k+2$:
\[
d'(1,k+1)=d(1,n+1)+d(n+1,n+2)+d(n+2,2)+... +d(k,k+1) = k+2
\]
The closeness is changed with $2 \cdot 2^{-k} -  2 \cdot 2^{-k-2}= 3 \cdot 2^{-k-1}$.
The same is the situation with the other $k-1$ distances:
$d(2n,k)$, $d(2n-1,k-1)$,...,$d(2n-k+2,2)$.

\medskip
The difference  and the residual closeness are:
\[
 \Delta_{2k+1} =\Delta_{2k-1} + 3k2^{-k-1}.
\]
\[
R(H_{3,4k+2})=C(H_{3,4k+2}) - \Delta_{2k-1} - 3k2^{-k-1}.
\]

We can prove now:
\begin{theorem}
The residual closeness of Harary graph $H_{3,2n}$ is:
\[
R(H_{3,4k}) =  C(H_{3,4k}) - 3 + (2k+3)2^{-k},
\]
\[
R(H_{3,4k+2}) =  C(H_{3,4k+2}) - 3 + 3(k+2)2^{-1-k}.
\]
\end{theorem}
\begin{proof}
From:\vspace*{-2mm}
\[
 \Delta_{2k+1} = \Delta_{2k-1} + 3k2^{-k-1}
\]
we receive:
\[
 \Delta_{2k+1} = \Delta_{5} + 3 \cdot 3 \cdot 2^{-4}+...+ 3k2^{-k-1}.
\]
\eject

\noindent Multiplying formula (1) by $\frac {3}{4}$ gives:
\begin{eqnarray*}
& 3(3 \cdot 2^{-4}+...+k2^{-1-k}) = 3 \left( 2^{-1} -(k+2)2^{-1-k} \right). &
\\[2pt]
&  \Delta_{2k+1} = \Delta_{5} + 3 \left( 2^{-1} -(k+2)2^{-1-k}\right). &
\\[2pt]
&  \Delta_{2k+1} = 3 -3(k+2)2^{-1-k}. &
\end{eqnarray*}
From $\Delta_{2k} = \Delta_{2k-1} + k2^{-k}$ we receive:
\[
 \Delta_{2k} = 3 -3(k+1)2^{-k} + k2^{-k} = 3 -(2k+3)2^{-k},
\]
which are exactly the formulae for the residual closeness of $H_{3,2n}$.
\end{proof}

\section{Residual closeness of $H_{5,2n}$}

Graph $H_{5,6}$ is a complete graph and deleting
any link will result in change of the distance from $1$ to $2$:
 $\Delta_1 =0.5$ and $R(H_{5,6}) =  C(H_{5,6}) - 0.5$.
Graph $H_{5,8}$ has also plenty of links and by deleting
any link, only one distance is changed  from $1$ to $2$:
$R(H_{5,8}) =  C(H_{5,8}) - \Delta_1 =  C(H_{5,8}) - 0.5$.

\medskip
For the bigger graphs we have to consider 3 cases:

\medskip
\textbf {Case 1 - Deleting link $(1,2)$:}

\noindent  Distance $d(1,2)$ is always changed from $1$ to $2$:
$d'(1,2) = d(1,3) +d(3,2)$.  No other distances are changed.

\medskip
\textbf {Case 2 - Deleting link $(1,n+1)$:}

\noindent  When $n>4$ distance $d(1,n+1)$ is changed from $1$ to $3$:
\[
d'(1,n+1) = d(1,2) +d(2,n+2) + d(n+2,n+1),
\]
 or the change is
$\Delta  = 2 \cdot 2^{-1} - 2 \cdot 2^{-3} = 0.75$.
 No other distances are changed.

\medskip
\textbf {Case 3 - Deleting link $(1,3)$:}

\noindent  By deleting link $(1,3)$,  distance $d(1,3)$ is changed from $1$ to $2$ and this is the only changed distance when  $n \le 6$.
Hence we receive $\Delta = 0.75$ and :
\[
R(H_{5,10}) =  C(H_{5,10}) - 0.75,  \quad  R(H_{5,12}) =  C(H_{5,12}) - 0.75.
\]
\noindent  When $n=7$, other distances start changing. Not only $d(1,3)$ is changed
from 1 to 2, but also $d(1,5)$ and $d(3,13)$ are changed from 2 to 3.
The residual closeness is:
\[
R(H_{5,14})=C(H_{5,14}) - 0.5 - 2(2 \cdot 2^{-2} -  2 \cdot 2^{-3})= C(H_{5,14}) - 1.
\]
The difference between the closeness and the residual closeness,
when $n=8,9,10$, is also $\Delta_2 =1.0$.
Now we can prove:

\begin{theorem}
The residual closeness of Harary graph $H_{5,2n}$ is:
\[
R(H_{5,2n}) =  C(H_{5,2n}) - 2 + (k+2)2^{-k},
\]
where $k = \lfloor \frac{n+1}{4} \rfloor$ and $n \ge 7$.
\end{theorem}

\begin{proof}
When $n=4k-1$, not only the previous distances are changed,
but new $k$ distances are changed from $k$ to $k+1$:
$d(1,1+2k)$, $d(2n-1,2k-1)$,..., $d(3,2n-2k+3)$.
The difference between the closeness and the residual closeness $\Delta_{k}$ is:
\[
\Delta_{k}=\Delta_{k-1} + k(2 \cdot 2^{-k} -  2 \cdot 2^{-k-1})=\Delta_{k-1} + k2^{-k}.
\]
\[
\Delta_{k}=\Delta_2 + 3 \cdot 2^{-3}+...+ k2^{-k}.
\]
The residual closeness is:
\[
R(H_{5,8k-2})=C(H_{5,8k-2}) - \Delta_{k}.
\]
The difference in closenesses $ \Delta_{k}$ is the same for $n=4k,4k+1,4k+2$.
Dividing formula (1) by 2 we receive:
\[
3 \cdot 2^{-3}+...+k2^{-k} = 1 - (k+2)2^{-k}.
\]
For the difference $\Delta_{k}$ we receive:
\[
\Delta_{k}=1 + 1 - (k+2)2^{-k} = 2 - (k+2)2^{-k},
\]
which proves the theorem.
\end{proof}

\section{Residual closeness of $H_{2p+1,2n}$}

We will follow the previous section.
When $n \in [p+1, 2p]$, by deleting any link, the distance is changed from $1$ to $2$: $\Delta_1=0.5$.

When $n \in [2p+1, 3p]$, by deleting link $(1,n+1)$,  distance $d(1,n+1)$ is changed from $1$ to $3$ and
$\Delta = 0.75$. This is the biggest decrement for $n$ in this range.

When $n > 3p$, by deleting link $(1,p+1)$,  distance
$d(1,p+1)$ is changed from 1 to 2. Also $d(1,2p+1)$ and $d(p+1,2n-p+1)$ are changed from 2 to 3. No other distances are changed when  $n \in [3p+1, 5p]$
and the decrement is: $\Delta_2 = 1$.

\medskip
In general, when $n=(2k-1)p+1$ and $k \ge2$, by deleting link $(1,p+1)$, not only the previous distances are changed, but new $k$ distances ($d(1,1+p.k)$,..., $d(p+1,2n-p(k-1)+1)$)
are changed from $k$ to $k+1$.
The differences is:
\[
\Delta_k = \Delta_{k-1} + k2^{-k}.
\]
Similarly to Theorem 4 we can prove:

\begin{theorem}
The residual closeness of Harary graph $H_{2p+1,2n}$ is:
\[
R(H_{2p+1,2n}) =  C(H_{2p+1,2n}) - 2 + (k+2)2^{-k},
\]
where $k = \lfloor \frac{n+p-1}{2p} \rfloor$, $p>1$,  and $n \ge 3p+1$.
\end{theorem}
\begin{proof}
The difference $\Delta_{k}$ is:
\[
\Delta_{k}=\Delta_2 + 3 \cdot 2^{-3}+...+ k2^{-k}.
\]
Dividing formula (1) by 2 we receive:
\[
3 \cdot 2^{-3}+...+k2^{-k} = 1 - (k+2)2^{-k}.
\]
Using $\Delta_2 = 1$, we receive:
\[
C(H_{2p+1,2n}) - R(H_{2p+1,2n}) = \Delta_{k}= 2 - (k+2)2^{-k},
\]
which proves the theorem.
\end{proof}

\section{Residual closeness of $H_{3,2n+1}$}

All vertices are connected to 3 other vertices,
 only vertex $n+1$ is connected to 4 vertices:
$1$, $n$, $n+2$, and $2n+1$.

Deleting any vertex of graph $H_{3,5}$  changes only this distance from $1$ to $2$
and the difference between the closeness and the residual closeness is
$\Delta_2= 0.5$.

When $n > 2$ we have to consider 4 cases.

\medskip
\textbf {Case 1 -  Deleting link $(1,n+1)$:}

 \noindent Distance $d(1,n+1)$ is changed from $1$ to $2$:
\[
d'(1,n+1) = d(1,2n+1) +d(2n+1,n+1) =2.
\]
When $n = 3$, this is the only changed distance and
the difference in closenesses is 0.5.
When $n > 3$, deleting link $(1,n+1)$ does not supply the residual closeness.

\medskip
\textbf {Case 2 -  Deleting link $(2,n+2)$:}

 \noindent When $n \ge 3$, distance $d(2,n+2)$ is changed from $1$ to $3$:
\[
d'(2,n+2) = d(2,1) + d(1,n+1) +d(n+1,n+2) =3.
\]
This is the only changed distance and
 the difference in closenesses is 0.75.

\medskip
\textbf {Case 3 -  Deleting link $(1,2)$:}

 \noindent When $n \ge 3$, distance $d(1,2)$ is changed from $1$ to $3$:
\[
d'(1,2) = d(1,n+1) +d(n+1,n+2) + d(n+2,2)=3.
\]
Distance $d(2,2n+1)$ is changed from $2$ to $3$ when $n = 3$:
$d'(2,2n+1) = d(2,n+2) +d(n+2,\linebreak n+1) + d(n+1,2n+1)=3$. The residual closeness,
when $n = 3$, is:
\[
\Delta_3=  2(2^{-1}-2^{-3}) + 2(2^{-2}-2^{-3}) = 1.
\]
\[
R(H_{3,7}) =  C(H_{3,7}) - \Delta_3 = C(H_{3,7}) -1.
\]
The only cases when deleting link $(1,2)$ supplies the
residual closeness are $n = 2, 3$.

\medskip
\textbf {Case 4 -  Deleting link $(n,n+1)$:}\smallskip

 \noindent \quad \quad \quad \textbf {A)} Distance $d(n,n+1)$ is changed from $1$ to $3$:
\[
d'(n,n+1) = d(n,2n) +d(2n,2n+1) + d(2n+1,n+1)=3.
\]
When $n=3$ this is the only changed distance.
The difference  is less than the difference
in case 3: $2(2^{-1}-2^{-3})=0.75 < \Delta_3 $.\smallskip

 \noindent \quad \quad \quad \textbf {B)} When $n \ge 4$ distance $d(1,n)$ is changed from $2$ to $3$ :
\[
d'(1,n) = d(1,2n+1) +d(2n+1,2n) + d(2n,n)=3.
\]
When $n=4$  distance $d(n-1,n+1)$ is changed from $2$ to $3$:
\[
d'(3,5) = d(3,7) +d(7,6) + d(6,5)=3.
\]
When $n=4$ distance $d(n,n+2)$ is also changed from $2$ to $3$:
\[
d'(4,6) = d(4,8) +d(8,7) + d(7,6)=3.
\]
These are the only changed distances when $n=4$ and:
\[
\Delta_4 = 2(2^{-1}-2^{-3}) + 3.2(2^{-2}-2^{-3})=1.5.
\]
\[
R(H_{3,9}) =  C(H_{3,9}) - \Delta_4 =  C(H_{3,9}) - 1.5.
\]

 \noindent \quad \quad \quad \textbf {C)} When $n>4$,  two of the changed (from 2 to 3) distances in subcase B have bigger changes (from 2 to 4).
Distance $d(n-1,n+1)$ is changed from $2$ to $4$:
\[
d'(n-1,n+1) = d(n-1,2n-1) +d(2n-1,2n) + d(2n,2n+1)+ d(2n+1,n+1).
\]
Distance $d(n,n+2)$ is also changed  from $2$ to $4$:
\[
d'(n,n+2) = d(n,2n) +d(2n,2n+1) + d(2n+1,n+1)+ d(n+1,n+2).
\]

These are the only changes when $n=5$ and:
\[
\Delta_5 = 2(2^{-1}-2^{-3}) + 2(2^{-2}-2^{-3})+ 2 \cdot 2(2^{-2}-2^{-4})=1.75.
\]
\[
R(H_{3,11}) =  C(H_{3,11}) - \Delta_5 =  C(H_{3,11}) - 1.75.
\]

 \noindent \quad \quad \quad \textbf {D)} In general, when $n \ge 2p$, new $p-1$ distances $d(1,n-p+2)$,$d(2,n-p+3)$,...,$d(p-1,n)$
 are changed from $p$  to $p+1$, e.g. from path $1,n+1,n, n-1,...,n-p+2$
to path $1,2n+1,2n, n, \linebreak n-1,...,n-p+2$.

\medskip
When $n=2p$ another $p$ distances $d(n-p+1,n+1)$,$d(n-p+2,n+2)$,...,$d(n-1,n+p+1)$
 are changed from $p$  to $p+1$ . e.g. from path $n+1,n, n-1,...,n-p+1$
to path $n+1,n+2,...,2n-p+1,n-p+1$. The distance between vertices $n+1 = 2p+1$ and $2n-p+1 = 3p+1$ is equal to $p$.
These are the only new changes when $n=2p$ and:
\[
\Delta_{2p} = \Delta_{2p-1} +(p-1)2(2^{-p}-2^{-p-1})+ p2(2^{-p}-2^{-p-1})=\Delta_{2p-1} + (2p-1) 2^{-p}.
\]
\[
R(H_{3,4p+1}) =  C(H_{3,4p+1}) - \Delta_{2p}.
\]

 \noindent \quad \quad \quad \textbf {E)} When $n>2p$ the $p$ distances $d(n-p+1,n+1)$,$d(n-p+2,n+2)$,...,$d(n-1,n+p+1)$ from subcase D
 are changed from $p$  to $p+2$, e.g. from path $n+1,n, n-1,...,n-p+1$
to path $n+1,2n+1,2n, n, n-1,...,n-p+1$.

\medskip
These are the only new changes when $n=2p+1$ and:
\[
\Delta_{2p+1} = \Delta_{2p} - p2(2^{-p}-2^{-p-1}) + p2(2^{-p}-2^{-p-2})  =\Delta_{2p} + p2^{-p-1}.
\]
\[
R(H_{3,4p+3}) =  C(H_{3,4p+3}) - \Delta_{2p+1}.
\]
Now we can prove:
\begin{theorem}
The residual closeness of Harary graph $H_{3,2n+1}$ is:
\[
R(H_{3,4p+1}) =  C(H_{3,4p+1}) - 4 + (3p+4)2^{-p},
\]
\[
R(H_{3,4p+3}) =  C(H_{3,4p+3}) - 4 + (5p+8)2^{-p-1},
\]
where  $p > 1$.
\end{theorem}

\begin{proof}

\vspace*{-9mm}
\[
\Delta_{2p} =\Delta_{2p-1} + (2p-1) 2^{-p}=\Delta_{2p-2} + (p-1)2^{-p} + (2p-1) 2^{-p}
\]
\begin{equation}
\label{eq2}
\Delta_{2p} = \Delta_{2p-2} + (3p-2) 2^{-p} .
\end{equation}

Formula (1) for $k=p$, divided by 2, becomes:
\[
3 \cdot 2^{-3}+...+p2^{-p} = 1 - (p+2)2^{-p}.
\]
Formula (1) for $k=p-1$ divided by 2, becomes:
\[
3 \cdot 2^{-3}+...+ 2(p-1)2^{-p} = 1 - (p+1)2^{1-p}.
\]
Adding both equation we receive:
\begin{equation}
\label{eq3}
6 \cdot 2^{-3}+...+(3p-2)2^{-p} = 2 -(3p+4)2^{-p}.
\end{equation}
The first items of the sum for $\Delta_{2p}$ are not added in the formula above.
To determine linear component L (the first items of the sum) we use:
\[
1.5 = \Delta_4 = L+2-10 \cdot 2^{-2}=L - 0.5,
\]
or $L =2$. Then the difference $\Delta_{2p}$ becomes:
\[
\Delta_{2p} = 4 - (3p+4)2^{-p}.
\]
For the next difference $\Delta_{2p+1}$ we receive:
\[
\Delta_{2p+1} =\Delta_{2p} + p2^{-p-1} = 4 - (5p+8)2^{-p-1},
\]
which proves the theorem.
\end{proof}

\section{Residual closeness of $H_{5,2n+1}$}

 \noindent \quad \quad \quad \textbf {A)} Deleting any link  $(i,j)$  of graph $H_{5,7}$ changes distance $d(i,j)$ from $1$ to $2$.
The same is the situation with graph $H_{5,9}$.
Hence:
\[
R(H_{5,2n+1})= C(H_{5,2n+1}) - 0.5, \quad when \quad n=3,4.
\]
 \noindent \quad \quad \quad \textbf {B)} For graph $H_{5,11}$, deleting link  $(2,n+2)$  changes distance $d(2,n+2)$ from $1$ to $3$.
Deleting a link, connecting nodes with closer numbers, like  $(1,2)$ or $(1,3)$, changes the distance from $1$ to $2$. The same change in the distance (from $1$ to $2$) causes deleting link $(1,n+1)$. Hence :
\[
R(H_{5,11})= C(H_{5,11}) - 2(2^{-1}-2^{-3}) = C(H_{5,11})  - 0.75.
\]

 \noindent \quad \quad \quad \textbf {C)} For graph $H_{5,13}$, deleting link  $(1,2n)$  changes  distance $d(1,2n)$ from $1$ to $2$
and  distances $d(1,2n-2)$ and $d(3,2n)$ from $2$ to $3$:
\[
R(H_{5,13})= C(H_{5,13}) - 2(2^{-1}-2^{-2}) - 2.2(2^{-2}-2^{-3}) = C(H_{5,13})  - 1.
\]

 \noindent \quad \quad \quad \textbf {D)} For graph $H_{5,2n+1}$, $n>6$, deleting link  $(n,n+2)$  changes  distance $d(n,n+2)$ from $1$ to $2$:
$d'(n,n+2) = d(n,n+1) +d(n+1,n+2)$. Distance $d(2,n)=d(2,n+2)+d(n+2,n)$ is changed
from $2$ to $3$: $d'(2,n) = d(2,n+2) +d(n+2,n+1)+d(n+1,n)$. The same is for distance  $d(n+2,2n)$.
Distance $d(n-2,n+2)=d(n-2,n)+d(n,n+2)$ is also changed
from $2$ to $3$: $d'(n-2,n+2) = d(n-2,n) +d(n,n+1)+d(n+1,n+2)$. The same is for distance  $d(n,n+4)$.
No other distance is changed when $n=7, 8, 9, 10$ and:
\[
\Delta_2 = 2(2^{-1}-2^{-2}) + 4.2(2^{-2}-2^{-3}) = 1.5,
\]
\[
R(H_{5,2n+1})= C(H_{5,2n+1})  - 1.5, \quad when \quad n=7,8,9,10.
\]

 \noindent \quad \quad \quad \textbf {E)} In general,
when $n=2k+1$, $k \in \{4p-1,4p,4p+1,4p+2 \}$,  deleting link
$(n,n+2)$ of graph $H_{5,2n+1}$,
in addition to the previous changed distances,
$3p-2$ distances are changed from $p$ to $p+1$. The change in closeness
$\Delta_{p}= C(H_{5,2n+1}) - R(H_{5,2n+1})$ is:
\begin{equation}
\label{eq4}
\Delta_{p} = \Delta_{p-1} + (3p-2).2(2^{-p}-2^{-p-1}) = \Delta_{p-1} + (3p-2)2^{-p}.
\end{equation}

We can prove now:
\begin{theorem}
The residual closeness of Harary graph $H_{5,2n+1}$ is:
\[
R(H_{5,2n+1}) =  C(H_{5,2n+1}) - 4 + (3p+4)2^{-p},
\]
where $p = \lfloor \frac{n+1}{4} \rfloor$ and $p > 1$.
\end{theorem}

\begin{proof}
Formula (4) is the same as formula (2) from Theorem 6. Using formula (3) from Theorem 6, we determine linear component $L$:
\[
1.5 =  \Delta_2 = L+2-10 \cdot 2^{-2}=L- 0.5,
\]
or $L =2$. Then the difference $\Delta_{p}$ becomes:
\[
\Delta_{p} = 4 - (3p+4)2^{-p},
\]
which proves the theorem.
\end{proof}

\section{Residual closeness of $H_{2m+1,2n+1}$}

We will consider the cases $m>2$ similar to $H_{5,2n+1}$.
The differences in closenesses of Harary graphs $H_{2m+1,2n+1}$ for the smaller
numbers $n$ are: $\Delta = 0.5$, when  $m+1 \le n \le 2m$;
$\Delta = 0.75$, when  $2m+1 \le n < 3m$; and
$\Delta = 1$, when  $n = 3m$.

\medskip
When  $n = 3m+1$, deleting link  $(n,n+m)$ of graph $H_{2m+1,6m+3}$,
changes distance $d(n,n+m)$ from $1$ to $2$ and $4$ more distances
($d(m,n)$, $d(n+m,2n)$, $d(n-m,n+m)$, and $d(n,n+2m)$) from $2$ to $3$. The difference in closenesses $\Delta_2$ is:
\[
\Delta_2 = 2(2^{-1}-2^{-2})+  4 \cdot 2(2^{-2}-2^{-3}) = 0.5 +1 = 1.5,
\]
\[
C(H_{2m+1,6m+3}) - R(H_{2m+1,6m+3}) =  \Delta_2 = 1.5.
\]

In general, when $n \in \{m(2p-1)+1,m(2p-1)+2,...,m(2p+1) \}$,  deleting link  $(n,n+m)$ of graph $H_{2m+1,2n+1}$,
in addition to the previous changed distances, new
$3p-2$ distances are changed from $p$ to $p+1$. The difference in closenesses
$\Delta_{p}$ is:
\begin{equation}
\label{eq5}
\Delta_{p} = \Delta_{p-1} + (3p-2)2(2^{-p}-2^{-p-1}) = \Delta_{p-1} + (3p-2)2^{-p}.
\end{equation}

We can prove now:
\begin{theorem}
The residual closeness of Harary graph $H_{2m+1,2n+1}$ is:
\[
R(H_{2m+1,2n+1}) =  C(H_{2m+1,2n+1}) - 4 + (3p+4)2^{-p},
\]
where $p = \lfloor \frac{n+m-1}{2m} \rfloor$ and $p > 1$.
\end{theorem}

\begin{proof}
Formula (5) is the same as formulae (2) and (4).
Similarly to the proof of Theorem 6 we have:
\[
\Delta_p = L+2 - (3p+4)2^{-p},
\]
where $L$ is a linear component, corresponding to the first terms of the sum of $\Delta_p$.
Using $\Delta_2 = 1.5$, we determine $L$:
\[
1.5 =  \Delta_2 = L+2-10 \cdot 2^{-2}=L- 0.5,
\]
or $L =2$. Then the difference $\Delta_{p}$ becomes:
\[
\Delta_{p} = 4 - (3p+4)2^{-p},
\]
which proves the theorem.
\end{proof}

\section{Conclusion}
The link residual closeness is one of the most sensitive indicators for robustness of networks. In this article we consider Harary graphs $H_{k,n}$.
The residual closeness of $H_{2,n}$ (cycle graph) is supplied by path graph $P_n$ (both with known closenesses).
When $k > 2$ we have calculated the difference between the closeness and
the link residual closeness of Harary graphs $H_{k,n}$.

\section{Appendix A. Proof of Formula 1.}

\begin{proof}
We start with:
\[
Y=X + X^2 + X^3+...+X^k.
\]
\[
Y(1-X)=X -X^{k+1},
\]
or:
\[
X + X^2 + X^3+...+X^k = \frac {X -X^{k+1}}{(1-X)}.
\]
Differentiating both sides of equation, we receive:
\[
1 + 2 \cdot X^1 + 3 \cdot X^2+...+kX^{k-1} = \frac {1 -(k+1)X^{k}}{(1-X)} + \frac {X -X^{k+1}}{(1-X)^2}.
\]
Replacing $X$ with $\frac {1}{2}$ we receive:
\[
1 + 2 \cdot 2^{-1}+ 3 \cdot 2^{-2}+...+k2^{1-k} = \frac {1 -(k+1)2^{-k}}{2^{-1}} + \frac {2^{-1} -2^{-k-1}}{2^{-2}}.
\]
\[
1 + 2 \cdot 2^{-1}+ 3 \cdot 2^{-2}+...+k2^{1-k} = 4 -2^{1-k} -(k+1)2^{1-k}.
\]
\[
3 \cdot 2^{-2}+...+k2^{1-k} = 2 -(k+2)2^{1-k},
\]
which is exactly Formula (1).
\end{proof}

\end{document}